\newcommand{\longversion}[1]{#1}
\newcommand{\shortversion}[1]{}

\documentclass[runningheads]{llncs}
\usepackage[T1]{fontenc}
\usepackage{graphicx}
\usepackage{hyperref}
\usepackage{amssymb,amsmath}
\usepackage{color}
\usepackage{makecell}
\usepackage{mathtools}
\usepackage{bbm}
\usepackage{setspace}

\usepackage[usenames,svgnames,dvipsnames]{xcolor}

\usepackage{tikz}
\usetikzlibrary{arrows,positioning}

\newdimen\prevdp
\def\leftlabel#1{\noalign{\prevdp=\prevdepth
   \kern-\prevdp\nointerlineskip\vbox to0pt{\vss\hbox{\ensuremath{#1}}}\kern\prevdp}}

\usepackage{url}

\usepackage{enumerate}

\usepackage{xspace}
\newcommand{\PPAD}{\ensuremath{\mathsf{PPAD}}\xspace}
\newcommand{\NP}{\ensuremath{\mathsf{NP}}\xspace}
\newcommand{\NPC}{\ensuremath{\mathsf{NP}}\text{-complete}\xspace}
\newcommand{\NPH}{\ensuremath{\mathsf{NP}}\text{-hard}\xspace}
\newcommand{\PNPH}{para-\ensuremath{\mathsf{NP}\text{-hard}}\xspace}
\newcommand{\el}{\ensuremath{\ell}\xspace}

\newcommand{\WOH}{\ensuremath{\mathsf{W[1]}}-hard\xspace}

\newcommand{\FPT}{\ensuremath{\mathsf{FPT}}\xspace}

\newcommand{\tsat}{\ensuremath{(3,\text{B}2)}-{\sc SAT}\xspace}

\let\oldlambda\lambda
\renewcommand{\lambda}{\ensuremath{\oldlambda}\xspace}
\let\oldalpha\alpha
\renewcommand{\alpha}{\ensuremath{\oldalpha}\xspace}
\let\oldDelta\Delta
\renewcommand{\Delta}{\ensuremath{\oldDelta}\xspace}

\newcommand{\YES}{{\sc yes}\xspace}

\newcommand{\yes}{{\sc yes}\xspace}
\newcommand{\no}{{\sc no}\xspace}
\newcommand{\true}{\text{{\sc true}}\xspace}
\newcommand{\false}{\text{{\sc false}}\xspace}

\newcommand{\CC}{\ensuremath{\mathcal C}\xspace}

\newcommand{\EE}{\ensuremath{\mathcal E}\xspace}

\newcommand{\GG}{\ensuremath{\mathcal G}\xspace}
\newcommand{\HH}{\ensuremath{\mathcal H}\xspace}
\newcommand{\II}{\ensuremath{\mathcal I}\xspace}

\newcommand{\RR}{\ensuremath{\mathcal R}\xspace}
\renewcommand{\SS}{\ensuremath{\mathcal S}\xspace}
\newcommand{\TT}{\ensuremath{\mathcal T}\xspace}

\newcommand{\VV}{\ensuremath{\mathcal V}\xspace}

\newcommand{\XX}{\ensuremath{\mathcal X}\xspace}

\newcommand{\EB}{\ensuremath{\mathbb E}\xspace}
\newcommand{\NB}{\ensuremath{\mathbb N}\xspace}
\newcommand{\RB}{\ensuremath{\mathbb R^+}\xspace}

\usepackage{nicefrac}

\newcommand{\eps}{\ensuremath{\varepsilon}\xspace}
\renewcommand{\epsilon}{\eps}

\newcommand{\ignore}[1]{}

\newcommand{\pr}{\ensuremath{\prime}}
\newcommand{\prr}{\ensuremath{{\prime\prime}}}

\renewcommand{\leq}{\leqslant}
\renewcommand{\geq}{\geqslant}
\renewcommand{\ge}{\geqslant}

\usepackage{enumitem}
\setlist[enumerate]{labelwidth=!, labelindent=0pt}

\usepackage{cleveref}

\crefname{theorem}{Theorem}{\bf Theorems}
\crefname{observation}{Observation}{\bf Observations}
\crefname{lemma}{Lemma}{\bf Lemmata}
\crefname{corollary}{Corollary}{\bf Corollaries}
\crefname{proposition}{Proposition}{\bf Propositions}
\crefname{definition}{Definition}{\bf Definitions}
\crefname{claim}{Claim}{\bf Claims}
\crefname{reductionrule}{Reduction rule}{\bf Reduction rules}

\usepackage{rotating}
\usepackage{amssymb}
\usepackage{latexsym}
\usepackage{epsfig}
\usepackage{xcolor}
\usepackage{url}
\usepackage{fancyhdr}

\usepackage{caption}
\usepackage{subcaption}
\usepackage{array}
\usepackage{multirow}
\usepackage{enumerate}
\usepackage{pdflscape}
\usepackage{amsmath,graphicx}
\usepackage{epstopdf}
\usepackage{float}
\floatstyle{ruled}
\newfloat{Algorithms}{!thb}{lop}
\floatname{Algorithms}{Algorithm}

\allowdisplaybreaks[1]

\usepackage{mathtools}

\usepackage{algorithmic}
\usepackage{algorithm}

\usepackage{comment}
\begin{document}

\title{On Binary Networked Public Goods Game with Altruism}
%
%
\author{Arnab Maiti\inst{1}\orcidID{0000-0002-9142-6255} \and
Palash Dey\inst{2}\orcidID{0000-0003-0071-9464}}
\authorrunning{Maiti et al.}
%
\institute{University of Washington 
\email{arnabm2@uw.edu\inst{1}}, Indian Institute of Technology Kharagpur \email{palash.dey@cse.iitkgp.ac.in\inst{2}}}
%


\maketitle              

\begin{abstract}
In the classical Binary Networked Public Goods (BNPG) game, a player can either invest in a public project or decide not to invest. Based on the decisions of all the players, each player receives a reward as per his/her utility function. However, classical models of BNPG game do not consider altruism which players often exhibit and can significantly affect equilibrium behavior. Yu et al.~\cite{ijcai2021-69} extended the classical BNPG game to capture the altruistic aspect of the players. We, in this paper, first study the problem of deciding the existence of a Pure Strategy Nash Equilibrium (PSNE) in a BNPG game with altruism. This problem is already known to be \NPC. We complement this hardness result by showing that the problem admits efficient algorithms when the input network is either a tree or a complete graph. We further study the Altruistic Network Modification problem, where the task is to compute if a target strategy profile can be made a PSNE by adding or deleting a few edges. This problem is also known to be \NPC. We strengthen this hardness result by exhibiting intractability results even for trees. A perhaps surprising finding of our work is that the above problem remains \NPH even for bounded degree graphs when the altruism network is undirected but becomes polynomial-time solvable when the altruism network is directed. We also show some results on computing an MSNE and some parameterized complexity results. In summary, our results show that it is much easier to predict how the players in a BNPG game will behave compared to how the players in a BNPG game can be made to behave in a desirable way.
\end{abstract}
\section{Introduction}

In a binary networked public goods (in short BNPG) game, a player can either decide to invest in a public project or decide not to invest in it. Every player however incurs a cost for investing. Based on the decision of all the players, each player receives a reward as per his/her externality function. The net utility is decided based on the reward a player receives and the cost a player incurs. Usually, the externality function and cost of investing differ for every player, making the BNPG game heterogeneous. In some scenarios, the externality function and cost of investing can be the same for every player, making the BNPG game fully homogeneous. Many applications of public goods, for example, wearing a mask \cite{maskpublic}, getting vaccinated \cite{brito1991externalities}, practicing social distancing \cite{cato2020social}, reporting crimes etc., involve binary decisions. Such domains can be captured using BNPG game. A BNPG game is typically modeled using a network of players which is an undirected graph \cite{galeotti2010network}.

We also observe that there are some societies where few people wear masks and/or get themselves vaccinated, and there are some other societies where most people wear masks and get themselves vaccinated \cite{nomask,nomask1}. This can be attributed to differences in altruistic behavior among various societies \cite{bir2021social,maria2021altruism}. In an altrusitic society, people consider their as well as their neighbors' benefit to take a decision. For example, young adults may wear a mask not only to protect themselves but also to protect their elderly parents and young children at home. Altruism can be modeled using an altruistic network which can be either an undirected graph or a directed graph \cite{ijcai2021-69}. Symmetric altruism (respectively asymmetric altruism) occurs when the altruistic network is undirected (respectively directed). The utility that a player receives depends on both the input network and the (incoming) incident edges in the altruistic network.

We study the BNPG game with altruism for two different problem settings. First, we look at the problem of deciding the existence of Pure Strategy Nash Equilibrium (PSNE) in the BNPG game with altruism. In any game, determining a PSNE is an important problem as it allows a social planner to predict the behaviour of players in a strategic setting and make appropriate decisions. It is known that deciding the existence of PSNE in a BNPG game (even without altruism) is NP-Complete \cite{yu2020computing}. This paper mainly focuses on deciding the existence of PSNE in special networks like trees, complete graphs, and graphs with bounded circuit rank. The circuit rank of an undirected graph is the minimum number of edges that must be removed from the graph to make it acyclic.

In the second problem setting, also known as Altruistic Network Modification (in short ANM), we can add or delete an edge from the altruistic network, and each such operation has a non-negative cost associated with it. The aim here is to decide if a target strategy profile can be made a PSNE by adding or deleting edges with certain budget constraints. This problem was first studied by \cite{ijcai2021-69} where they showed that ANM is an NP-Complete problem. This problem enables policymakers to strategically run campaigns to make a society more altruistic and achieve desirable outcomes like everyone wearing a mask and getting vaccinated.  This paper mainly focuses on ANM in sparse input networks like trees and graphs with bounded degree.

\subsection{Contribution} 
\begin{center}
	\begin{tabular}{|m{2cm} |m{1.5cm}  |m{1.5cm} |m{1.5cm} |  }
		\hline
		Input graph type & PSNE existence & ANM symmetric altruism & ANM asymmetric altruism \\
		\hline
		Tree & \textbf{P} & \textbf{NP-hard} & \textbf{NP-hard} \\
		\hline
		Clique & \textbf{P} & \NPH ($\star$) & \NPH ($\star$) \\
		\hline
		Bounded degree & \NPH ($\star\star$) & \textbf{NP-hard} & \textbf{P}\\
		\hline
		Bounded circuit rank & \textbf{P} & \textbf{NP-hard} & \textbf{NP-hard}\\
		\hline
	\end{tabular}
	\captionof{table}{List of results (our results are in bold). PSNE existence results hold for both symmetric and asymmetric altruism.}
	\label{table:1}
\end{center}

We show that the problem of deciding the existence of PSNE in BNPG game with asymmetric altruism is polynomial-time solvable if the input network is either a tree~[\Cref{thm:1}], complete graph~[\Cref{thm:3}] or graph with bounded circuit rank~[\Cref{thm:2}]. Moreover, in \Cref{thm:1}, we formulated a non-trivial ILP (not the ILP that follows immediately from the problem definition) and depicted a greedy polynomial time algorithm~[\Cref{algo:1}] to solve it. This strengthens the tractable results for tree, complete graph and graph with bounded circuit rank in \cite{yu2020computing,maiti2020parameterized} as the previous results were depicted for BNPG games without altruism. Hence, existence of a PSNE can be efficiently decided in an intimately connected society where everyone knows others and thus the underlying graph is connected, and for sparsely connected society where the circuit rank could be low. However, the problem is open for graphs with bounded treewidth, and it is known that the problem even without altruism is \WOH for the parameter treewidth ~\cite{maiti2020parameterized}. The problem of deciding the existence of PSNE in BNPG game even without altruism is known to be \NPC \cite{yu2020computing}. A natural but often under-explored question here is if an MSNE can be computed efficiently. We show that computing an MSNE in BNPG game with symmetric altruism is \PPAD-hard~[\Cref{thm:4}].

ANM with either asymmetric or symmetric altruism is known to be \NPC when the input network is a clique \cite{ijcai2021-69}. We complement this by showing that ANM with either asymmetric or symmetric altruism is \NPC even when the input network is a tree (the circuit rank of which is zero) and the BNPG game is fully homogeneous~[\Cref{thm:5,thm:6}]. We also show that ANM with symmetric altruism is known to be \PNPH for the parameter maximum degree of the input network even when the BNPG game is fully homogeneous and the available budget is infinite~[\Cref{thm:9}].  However, with asymmetric altruism, the problem is \FPT for the parameter maximum degree of the input network~[\Cref{thm:8}]. To show this result, we designed an $O(2^{n/2})$ time binary search based algorithm~[\Cref{algo:2}] for Minimum Knapsack problem. We are the first to provide an algorithm better than $O(2^n)$ time for Minimum Knapsack problem to the best of our knowledge.

In summary, our paper provides a more fine-grained complexity theoretic landscape for deciding if a PSNE exists in a BNPG game with altruism and the ANM problem which could be of theoretical as well as practical interest.
We summarize all the main results (including that of prior work) in \Cref{table:1}. There ($\star$) denotes the results from \cite{ijcai2021-69} and ($\star\star$) denotes the results from \cite{maiti2020parameterized}. All the hardness results in the table except for complete graph hold even for fully homogeneous BNPG game. We observe that the PSNE existence problem admits efficient algorithm for many settings compared to ANM. This seems to indicate that enforcing a PSNE is computationally more difficult than finding a PSNE.



\subsection{Related Work}
Our work is related to \cite{yu2020computing} who initiated the study of computing a PSNE in BNPG games. Their results were strengthened by \cite{maiti2020parameterized} who studied the parameterized complexity of deciding the existence of PSNE in BNPG game. Recently, \cite{papadimitriou2021public} studied about public goods games in directed networks and showed intractibility for deciding the existence of PSNE and for finding MSNE. Our work is also related to \cite{ijcai2021-69} who intiated the study of Altruistic Network Modification in BNPG game. \cite{ledyard20202,meier2008windfall} also discussed different ways to capture altruism. In the non-altruistic setting, \cite{kempe2020inducing} worked on modifying networks to induce certain classes of equilibria in BNPG game. Our work is part of graphical games where the fundamental question is to determine the complexity of finding an equilibrium \cite{daskalakis2009complexity,elkind2006nash,gottlob2005pure}. Our model is also related to the best-shot games \cite{dall2011optimal} as it is a special case of BNPG game.~\cite{galeotti2010network,levit2018incentive,manshadi2009supermodular,komarovsky2015efficient} also discussed some important variations of graphical games.

\section{Preliminaries}\label{sec:prelim}

Let $[n]$ denote the set $\{1,\ldots,n\}$. Let $\GG=(\VV,\EE)$ be an input network with $n$ vertices (each denoting a player). The input network is always an undirected graph. Let $\HH=(\VV,\EE^\pr)$ be an altruistic network on the same set of $n$ vertices. The altruistic network can be directed or undirected graph. An undirected edge between $u,v\in \VV$ is represented by $\{u,v\}$. Similarly a directed edge from $u$ to $v$ is represented by $(u,v)$. $N_v$ is the set of all neighbours (resp. out-neighbours) of the vertex $v$ in an undirected (resp. directed) altrusitic network \HH. Note that $N_v$ is a subset of neightbours of $v$ in \GG. A Binary Networked Public Goods (BNPG) game with asymmetric (resp. symmetric) altruism can be defined on the input graph \GG and the directed (resp. undirected) altruistic network \HH as follows. We are given a set of players \VV, and the strategy set of every player in \VV is $\{0,1\}$. For a strategy profile $\textbf{x}=(x_v)_{v\in\VV}\in\{0,1\}^{|\VV|}$, let $n_v=|\{u\in\VV:\{u,v\}\in\EE, x_u=1\}|$. In this paper, we will be using playing 1 (resp. 0), investing (resp. not investing) and strategy $x_v=1$ (resp. $x_v=0$) interchangeably. Now the utility $U_v(\textbf{x})$ of player $v\in\VV$ is defined as follows.
\begin{equation*}
	U_v(\textbf{x})=g_v(x_v+n_v)+a\sum_{u\in N_v}g_u(x_u+n_u)-c_v\cdot x_v
\end{equation*}
where $g_v:\NB\cup\{0\}\rightarrow\RB\cup\{0\}$ is a non-decreasing externality function in $x$ and $a, c_v\in\RB\cup\{0\}$ are constants. $c_v$ can also interpreted as the cost of investing for player $v$. We denote a BNPG game with altruism by $(\GG=(\VV,\EE),\HH=(\VV,\EE^\pr),(g_v)_{v\in\VV},(c_v)_{v\in\VV},a)$. We also define $\Delta g(x)=g(x+1)-g(x)$ where $x\in \NB\cup\{0\}$.
In this paper, we study a general case of BNPG game called {\em heterogeneous} BNPG game where every player $v\in\VV$ need not have the same externality function $g_v(.)$ and constant $c_v$. If nothing is mentioned, by BNPG game, we are referring to a heterogeneous BNPG game. In this paper, we also study a  special case of BNPG game called {\em fully homogeneous} BNPG game where $g_v=g$ for all $v\in\VV$ and $c_v=c$ for all $v\in \VV$.

In this paper, we mainly focus on {\em pure-strategy Nash Equilibrium (PSNE)}. A strategy profile $\textbf{x}=(x_v)_{v\in\VV}$ is said to be a PSNE of a BNPG game with altruism if the following holds true for all $v\in\VV$ and for all $x^\pr_v\in\{0,1\}$
\[  U_v(x_v,x_{-v}) \ge U_v(x^\pr_v, x_{-v})\]
where $x_{-v}=(x_u)_{u\in \VV\setminus \{v\}}$.

In this paper, we also look at {\em $\varepsilon$-Nash Equilibrium}. Let $\Delta_v$ be a distribution over that strategy set $\{0,1\}$. We define Supp$(\Delta_v)$ to be the support of the distribution $\Delta_v$, that is,  Supp$(\Delta_v)=\{x_v:x_v\in\{0,1\},\Delta_v(x_v)>0\}$  where $\Delta_v(x_v)$ denotes the probability of choosing the strategy $x_v$ by player $v$. Now $(\Delta_v)_{v\in \VV}$ is an {\em $\varepsilon$-Nash Equilibrium} if the following holds true for all  $x^\pr_v\in\{0,1\}$, for all $x_v\in \text{Supp}(\Delta_v)$, for all $v\in\VV$:
\[ \EB_{x_{-v}\sim\Delta_{-v}}[U_v(x_v,x_{-v})] \ge \EB_{x_{-v}\sim\Delta_{-v}} [U_v(x^\pr_v, x_{-v})] -\varepsilon\]
where $\Delta_{-v}=(\Delta_u)_{u\in \VV\setminus \{v\}}$.

\subsection{Altruistic Network Modification}
In this paper we study a special case of Altruistic Network Modification (ANM) which was also studied by \cite{ijcai2021-69}. If nothing is mentioned, by ANM, we are referring to the special case which we will now discuss. We are given a target profile $\textbf{x}^*$, BNPG game on an input graph \GG, an initial altruistic network \HH , a cost function $C(.)$ and budget $B$. In this setting, we can add or delete an edge $e$ from \HH and each such operation has a non-negative cost $C(e)$ associated with it. We denote an instance of ANM with altruism by  $(\GG=(\VV,\EE),\HH=(\VV,\EE^\pr),(g_v)_{v\in\VV},(c_v)_{v\in\VV},a, C(.),B,\textbf{x}^*)$. The aim of ANM with altruism is to add and delete edges in \HH such that $\textbf{x}^*$ becomes a PSNE and the total cost for adding and deleting these edges is atmost $B$. Note that if the altruism is asymmetric (resp. symmetric), then we can add or delete directed (resp. undirected) edges only. We are not allowed to add any edge between two nodes $u,v$ if $\{u,v\}\notin \EE$.
\subsection{Standard Definitions}
\shortversion{Due to space constraints, we refer the reader to the supplementary material for the definitions of Circuit Rank, \FPT and \PNPH.}
\longversion{
	\begin{definition}[Circuit Rank]\cite{maiti2020parameterized}
		Let the number of edges and number of vertices in a graph $\mathcal{G}$ be $m$ and $n$ respectively. Then circuit rank is defined to be $m-n+c$ ($c$ is the number of connected components in the graph). Note that circuit rank is not the same as feedback arc set.
	\end{definition}
	
	\begin{definition}[\FPT]\cite{maiti2021parameterized}
		A tuple $(x, k)$, where k is the parameter, is an instance of a parameterized problem. \emph{Fixed parameter tractability} (FPT) refers to solvability in time  $f(k) \cdot p(|x|)$ for a given instance $(x, k)$, where  $p$ is a polynomial in the input size $|x|$ and $f$ is an arbitrary computable function of $k$.
	\end{definition}
	\begin{definition}[\PNPH]\cite{maiti2021parameterized}
		We say a parameterized problem is \PNPH if it is \NPH even for some constant values of the parameter.
	\end{definition}
}
\section{Results for computing equilibrium}
In this section, we present the results for deciding the existence of PSNE and finding MSNE in BNPG game with altruism. {We have omitted few proofs. They are marked by ($\star$) and they are available in the appendix. }

\cite{yu2020computing} showed that the problem of checking the existence of PSNE in BNPG game without altruism is polynomial time solvable when the input network is a tree. We now provide a non-trivial algorithm to show that the problem of checking the existence of PSNE in BNPG game with asymmetric altruism is polynomial time solvable when the input network is a tree.
\longversion{
\begin{theorem}}
\shortversion{
\begin{theorem}[$\star$]}\label{thm:1}
The problem of checking the existence of PSNE in BNPG game with asymmetric altruism is polynomial time solvable when the input network is a tree.
\end{theorem}
\longversion{
\begin{proof}}
\shortversion{
\begin{proof}[Proof Sketch]}
For each player $v\in \VV$, let $d_v$ denote the degree of $v$. At each node $v$ with parent $u$, we maintain a table of tuples $(x_u,n_u,x_v,n_v)$ of valid configurations. A tuple $(x_u,n_u,x_v,n_v)$ is said to be a valid configuration if there exists a strategy profile $\textbf{x}^\pr=(x_v^\pr)_{v\in\VV}$ such that the following holds true:
\begin{itemize}
 \item $x_u^\pr=x_u$, $x_v^\pr=x_v$
\item The number of neighbours of $u$ and $v$ playing $1$ in $\textbf{x}^\pr$ is $n_u$ and $n_v$ respectively
\item None of the players in the sub-tree rooted at $v$ deviate from their strategy in $\textbf{x}^\pr$
\end{itemize}
Note that the root node $r$ doesn't have a parent. Hence, we consider an imaginary parent $p$ with $x_p=0$ and $g_p(x)=0$ for all $x\geq 0$. Hence if there is a tuple in the table of root node $r$ then we can conclude that there is a PSNE otherwise we can conclude that there is no PSNE.

\textbf{Leaf nodes:} We add a tuple $(x_u,n_u,x_v,n_v)$ to the table if $n_v=x_u$, $v$ does not deviate if it plays $x_v$ and $x_v\leq n_u\leq d_u+x_v-1$. Table for the leaf node can be clearly constructed in polynomial time.

\textbf{Non-leaf nodes:} For each tuple  $(x_u,n_u,x_v,n_v)$ we do the following. If there is no child $u^\pr$ of $v$ having a tuple of type $(x_v,n_v,x_{u^\pr},n_{u^\pr})$ in its table, then we don't add $(x_u,n_u,x_v,n_v)$ to the table of $v$. Similarly if $n_u> d_u+x_v-1$ or $n_u<x_v$, then we don't add $(x_u,n_u,x_v,n_v)$ to the table of $v$. Otherwise we do the following. Let $U_1$ be the set of children $u'$ of $v$ which have tuples of the type $(x_v,n_v, 1,n_{u'})$ in their table but don't have tuples of the type $(x_v,n_v, 0,n_{u'})$. Let $U_0$ be the set of children $u'$ of $v$ which have tuples of the type $(x_v,n_v, 0,n_{u'})$ in their table but don't have tuples of the type $(x_v,n_v, 1,n_{u'})$. Let $U$ be the set of children $u'$ of $v$ which have tuples of the type $(x_v,n_v, 1,n_{u'})$ and $(x_v,n_v, 0,n_{u'})$ in their table. First let us consider the case when $x_v=1$. Now for each $u'\in (U_1\cup U)\cap N_v$, we find the tuple  $(1,n_v, 1,n_{u'})$ in its table so that $a\cdot\Delta g_{u'}(n_{u'})$ is maximized and let this value be $y_{u'}$. Similarly for each $u'\in (U_0\cup U)\cap N_v$, we find the tuple  $(1,n_v, 0,n_{u'})$ in the table so that $a\cdot\Delta g_{u'}(n_{u'}-1)$ is maximized and let this value be $z_{u'}$. Also $\forall u'\in(U_1\cup U_0\cup U)\setminus N_v$, $y_{u'}=z_{u'}=0$. If $u\in N_v$ then $y_u= a\cdot \Delta g_u(x_u+n_u-1)$ otherwise $y_u=0$. Now we include the tuple $(x_u,n_u,x_v=1,n_v)$ in the table if the optimal value of the following ILP is at least $c_v-\Delta g_v(n_v)-y_u-\sum_{u'\in U_1}y_{u'}-\sum_{u'\in U_0}z_{u'}$.
\begin{align*}
\text{max} \quad
 &\sum_{u'\in U}(x_1^{u'}y_{u'}+x_0^{u'}z_{u'}) \\
\text{s.t.} \quad
&x_1^{u'}+x_0^{u'} = 1 \quad \forall u'\in U\\
&\sum_{u'\in U}x_1^{u'}=n_v-x_u-|U_1| \\
&x_1^{u'},x_0^{u'}\in \{0,1\} \quad \forall u'\in U
\end{align*}
The above ILP can be solved in polynomial time as follows. First sort the values $a_{u'}=|y_{u'}-z_{u'}|$ in non-increasing order breaking ties arbitrarily and order the vertices in $U$ as $\{u_1,\ldots,u_{|U|}\}$ as per this order, that is, $a_{u_i}\geq a_{u_j}$ if $i\leq j$. Then we traverse the list of values in non-increasing order and for the $u'$ corresponding to the value, we choose $x_1^{u'}=1$ if $y_{u'}>z_{u'}$ otherwise we choose $x_0^{u'}=1$. We do this until $\sum_{u'\in U}x_1^{u'} =n_v^\pr$ or $\sum_{u'\in U}x_0^{u'}=|U|-n_v^\pr$ where $n_v^\pr=n_v-x_u-|U_1|$. Remaining values are chosen in a way such that $\sum_{u'\in U}x_1^{u'} =n_v$ is satisfied. For a more detailed description, please see the \Cref{algo:1}.

Now we show the correctness. Consider an optimal solution $x^*$. Let $i$ be the smallest number such that $x_1^{u_i}=1$ as per our algorithm and in optimal solution it is $0$. Similarly let $i'$ be the smallest number such that $x_1^{u_{i'}}=0$ as per our algorithm and in optimal solution it is $1$. We now swap the values of the variables $x_1^{u_i}$ and $x_1^{u_{i'}}$ (resp. $x_0^{u_i}$ and $x_0^{u_{i'}}$) in the optimal solution without decreasing the value of the objective function. Let us assume that $i<i'$. Then it must be the case that $y_{u_i}>z_{u_i}$ otherwise $\forall j>i$, we will have $x_1^{u_{j}}=1$ as per our algorithm. Hence by swapping in the optimal solution, the value of the objective function increases by at least $a_{u_{i}}-a_{u_{i'}}$ which is a non-negative quantity.\shortversion{ Similar arugument exists for the case when $i'<i$.} \longversion{Similarly when $i'<i$, it must  be the case that $y_{u_i}\leq z_{u_i}$ otherwise $\forall j>i$, we will have $x_1^{u_{j}}=0$ as per our algorithm. Hence by swapping in the optimal solution, the value of the objective function increases by at least $a_{u_{i'}}-a_{u_{i}}$ which is a non-negative quantity.} By repeatedtly finding such indices $i,i'$ and then swaping the value of $x_1^{u_i}$ and $x_1^{u_{i'}}$ (resp. $x_0^{u_i}$ and $x_0^{u_{i'}}$) in the optimal solution leads to our solution.

\shortversion{An analogous procedure exists for the case where $x_v=0$. As mentioned earlier if there is any tuple in the table of the root $r$, then we conclude that there is a PSNE otherwise we conclude that there is no such PSNE.}
\longversion{An analogous procedure exists for the case where $x_v=0$. For each $u'\in (U_1\cup U)\cap N_v$, we find the tuple  $(0,n_v, 1,n_{u'})$ in the table so that $a\cdot\Delta g_{u'}(n_{u'}+1)$ is minimized and let this value be $y_{u'}$. Similarly for each $u'\in (U_0\cup U)\cap N_v$, we find the tuple  $(1,n_v, 0,n_{u'})$ in the table so that $a\cdot\Delta g_{u'}(n_{u'})$ is minimized and let this value be $z_{u'}$. Also $\forall u'\in(U_1\cup U_0\cup U)\setminus N_v$, $y_{u'}=z_{u'}=0$. If $u\in N_v$ then $y_u= a\cdot \Delta g_u(x_u+n_u)$ otherwise $y_u=0$. Now we include the tuple $(x_u,n_u,x_v=0,n_v)$ in the table if the optimal value of the following ILP is at most $c_v-\Delta g_v(n_v)-y_u-\sum_{u'\in U_1}y_{u'}-\sum_{u'\in U_0}z_{u'}$.
\begin{align*}
\text{min} \quad
 &\sum_{u'\in U}(x_1^{u'}y_{u'}+x_0^{u'}z_{u'}) \\
\text{s.t.} \quad
&x_1^{u'}+x_0^{u'} = 1 \quad \forall u'\in U\\
&\sum_{u'\in U}x_1^{u'}=n_v-x_u-|U_1| \\
&x_1^{u'},x_0^{u'}\in \{0,1\} \quad \forall u'\in U
\end{align*}
The above ILP can be solved in polynomial time as follows. First sort the values $a_{u'}=|y_{u'}-z_{u'}|$ in non-increasing order breaking ties arbitrarily and order the vertices in $U$ as $\{u_1,\ldots,u_{|U|}\}$ as per this order, that is, $a_{u_i}\geq a_{u_j}$ if $i\leq j$. Then we traverse the list in non-increasing order and for the $u'$ corresponding to the value, we choose $x_0^{u'}=1$ if $y_{u'}>z_{u'}$ otherwise we choose $x_1^{u'}=1$. We do this until $\sum_{u'\in U}x_1^{u'} =n_v^\pr$ or $\sum_{u'\in U}x_0^{u'}=|U|-n_v^\pr$ where $n_v^\pr=n_v-x_u-|U_1|$. Remaining values are chosen in a way such that $\sum_{u'\in U}x_1^{u'} =n_v$ is satisfied.

Now we show the correctness. Consider an optimal solution $x^*$. Let $i$ be the smallest number such that $x_1^{u_i}=1$ as per our algorithm and in optimal solution it is $0$. Similarly let $i'$ be the smallest number such that $x_1^{u_{i'}}=0$ as per our algorithm and in optimal solution it is $1$. We now swap the values of the variables $x_1^{u_i}$ and $x_1^{u_{i'}}$ (resp. $x_0^{u_i}$ and $x_0^{u_{i'}}$) in the optimal solution without increasing the value of the objective function. Let us assume that $i<i'$. Then it must be the case that $y_{u_i}\leq z_{u_i}$ otherwise $\forall j>i$, we will have $x_1^{u_{j}}=1$ as per our algorithm. Hence by swapping in the optimal solution, the value of the objective function decreases by at least $a_{u_{i}}-a_{u_{i'}}$ which is a non-negative quantity. Similarly when $i'<i$, it must  be the case that $y_{u_i}> z_{u_i}$ otherwise $\forall j>i$, we will have $x_1^{u_{j}}=0$ as per our algorithm. Hence by swapping in the optimal solution, the value of the objective function increases by at least $a_{u_{i'}}-a_{u_{i}}$ which is a non-negative quantity. By repeatedtly finding such indices $i,i'$ and then swapping the value of $x_1^{u_i}$ and $x_1^{u_{i'}}$ (resp. $x_0^{u_i}$ and $x_0^{u_{i'}}$) in the optimal solution leads to our solution.

As mentioned earlier if there is any tuple in the table of the root $r$, then we conclude that there is a PSNE otherwise we conclude that there is no such PSNE.}
\end{proof}

\begin{algorithm}[!htbp]
	\caption{ILP Solver} \label{algo:1}
	\begin{algorithmic}[1]
\STATE $\forall u' \in U$, $a_{u'}\gets |y_{u'}-z_{u'}|$.
\STATE Order the vertices in $U$ as $\{u_1,\ldots, u_{|U|}\}$ such that $\forall i,j\in[n]$, we have $a_{u_i}\geq a_{u_j}$ if $i\leq j$.
\STATE $\forall u' \in U$, $x_0^{u'}\gets 0$ and $x_1^{u'}\gets 0$
\STATE $n_v^\pr\gets n_v-x_u-|U_1|$
\FOR{ $i=1$ to $|U|$}
\IF{$\sum_{u'\in U}x_1^{u'} =n_v^\pr$}
\STATE $x_0^{u_i}\gets 1$ and $x_1^{u_i}\gets 0$
\ELSIF{$\sum_{u'\in U}x_0^{u'}=|U|-n_v^\pr$}
\STATE $x_0^{u_i}\gets 0$ and $x_1^{u_i}\gets 1$
\ELSE
\IF{$y_{u_i}>z_{u_i}$}
\STATE $x_0^{u_i}\gets 0$ and $x_1^{u_i}\gets 1$
\ELSE
\STATE $x_0^{u_i}\gets 1$ and $x_1^{u_i}\gets 0$
\ENDIF
\ENDIF
\ENDFOR
	\end{algorithmic}
\end{algorithm}

\begin{corollary}\label{cor:tree}
Given a BNPG game with asymmetric  altruism on a tree $\GG=(\VV,\EE)$, a set $\SS\subseteq\VV$ and a pair of tuples $(x_v^\pr)_{v\in\SS}\in\{0,1\}^{|\SS|}$ and $(n_v^\pr)_{v\in\VV^\pr}\in\{0,1,\ldots,n\}^{|\SS|}$, we can decide in polynomial time if there exists a PSNE $\textbf{x}^*=(x_v)_{v\in\VV}\in\{0,1\}^n$ for the BNPG game with asymmetric  altruism such that $x_v=x_v^\pr$ for every $v\in\SS$ and the number of neighbors of $v$ playing $1$ in $\textbf{x}^*$ is $n_v^\pr$  for every $v\in \SS$.
\end{corollary}
\longversion{
\begin{proof}
In the proof of Theorem 1, just discard those entries from the table of $u\in \SS$ which don't have $x_u=x_u^\pr$ and $n_u=n_u^\pr$.
\end{proof}
}
\cite{maiti2020parameterized} showed that the problem of checking the existence of PSNE in BNPG game without altruism is polynomial time solvable when the input network is a graph with bounded circuit rank. By using the algorithm in Theorem 1 as a subroutine and extending the ideas of \cite{maiti2020parameterized} to our setting, we show the following.

\begin{theorem}[$\star$]\label{thm:2}
The problem of checking the existence of PSNE in BNPG game with asymmetric altruism is polynomial time solvable when the input network is a graph with bounded circuit rank.
\end{theorem}

\cite{yu2020computing,maiti2020parameterized} showed that the problem of checking the existence of PSNE in BNPG game without altruism is polynomial time solvable when the input network is a complete graph. By extending their ideas to our setting, we show the following.

\begin{theorem}[$\star$]\label{thm:3}
The problem of checking the existence of PSNE in BNPG game with asymmetric altruism is polynomial time solvable when the input network is a complete graph.
\end{theorem}

The problem of deciding the existence of BNPG game with altruism where the atruistic network is empty is known to be NP-Complete\cite{yu2020computing}. Therefore, we look at the deciding the complexity of finding an $\varepsilon$-Nash equilibrium in BNPG game with symmetric altruism. We show that it is \PPAD-Hard. Towards that, we reduce from an instance of Directed public goods game which is known to be \PPAD-hard \cite{papadimitriou2021public}. In directed public goods game, we are given a directed network of players and the utility $U_u^\pr(x_v,x_{-v})$ of a player $v$ is $Y(x_v+n_v^{in})-p\cdot x_v$. Here $x_v\in\{0,1\}$, $n_v^{in}$ is the number of in-neighbours of $v$ playing $1$ and $Y(x)=0$ if $x=0$ and $Y(x)=1$ if $x>0$.

\begin{theorem}\label{thm:4}
Finding an $\varepsilon$-Nash equilibrium of the BNPG game with symmetric altruism is \PPAD-hard, for some constant $\varepsilon > 0$.
\end{theorem}
\begin{proof}
Let $(\GG(\VV,\EE),p)$ be an input instance of directed public goods game. Now we create an instance $(\GG^\pr=(\VV^\pr,\EE^\pr),\HH=(\VV^\pr,\EE^\prr),(g_v)_{v\in\VV^\pr},(c_v)_{v\in\VV^\pr},a)$ of BNPG game with symmetric altruism.
\begin{align*}
		\VV^\pr &= \{u_{in}:u\in\VV\}\cup\{u_{out}:u\in\VV\}\\
		\EE^\pr &= \{\{u_{in},u_{out}\}: u\in\VV^\pr\}\cup \{\{u_{out},v_{in}\}: (u,v)\in\EE\}\\
		\EE^\prr &= \{\{u_{in},u_{out}\}: u\in\VV^\pr\}
\end{align*}
Let the constant $a$ be $1$. $\forall u\in\VV$, $c_{u_{in}}=1+2\varepsilon$ and $c_{u_{out}}=p$. Now define the functions $g_w(.)$ as follows:
\[
	\forall u\in\VV, g_{u_{in}}(x) = \begin{cases}
	1 & x>0\\
	0 & \text{otherwise }
	\end{cases}\]
\[
	\forall u\in\VV, g_{u_{out}}(x) = 0\text{ } \forall x\geq 0\]
Now we show that given any $\varepsilon$-Nash equilibrium of the BNPG game with altruism ,  we can find an $\varepsilon$-Nash equilibrium of the directed public goods game in polynomial time. Let $(\Delta_u)_{u\in \VV^\pr}$ be an $\varepsilon$-Nash equilibrium of the BNPG game with symmetric altruism.

For all $v\in\{u_{in}:u\in\VV\}$, we have the following:
\begin{align*}
\mathbb{E}_{x_{-v}\sim\Delta_{-v}}[U_v(0,x_{-v})]-\varepsilon\geq -\varepsilon> -2\varepsilon
= 1- c_{u_{in}}
= \mathbb{E}_{x_{-v}\sim\Delta_{-v}}[U_v(1,x_{-v})]
\end{align*}
Hence $1$ can't be in the support of $\Delta_{v}$. Therefore $\forall u\in \VV$, $\Delta_{u_{in}}(0)=1$.

Now we show that $(\Delta_u)_{u\in\VV}$ is an $\varepsilon$-Nash equilibrium of the directed public goods game where $\Delta_u=\Delta_{u_{out}}$  $\forall u\in \VV$. Now consider a strategy profile $(x_v)_{v\in\VV^\pr}$ such that $\forall u\in \VV$, we have $x_{u_{in}}=0$. Now let $(x_v)_{v\in\VV}$ be a strategy profile such that $\forall u\in \VV$ we have $x_u=x_{u_{out}}$. Then we have the following:
\begin{align*}
U_{v_{out}}(x_{v_{out}},x_{-v_{out}})=g_{v_{in}}(n_{v_{in}})-p\cdot x_{v_{out}}
=Y(x_v+n_v^{in})-p\cdot x_v=U_u^\pr(x_u,x_{-u})
\end{align*}
Using the above equality and the fact that $\forall u\in \VV$, $\Delta_{u_{in}}(0)=1$, we have $\mathbb{E}_{x_{-u}\sim\Delta_{-u}}[U_u^\pr(x_u,x_{-u})]=\mathbb{E}_{x_{-u_{out}}\sim\Delta_{-u_{out}}}[U_{u_{out}}(x_{u_{out}},x_{-u_{out}})]$ where $x_u=x_{u_{out}}$. For all $u\in\VV$, for all  $x^\pr_u\in\{0,1\}$, for all $x_u\in \text{Supp}(\Delta_u)$,  we have the following:
\begin{align*}
\mathbb{E}_{x_{-u}\sim\Delta_{-u}}[U_u^\pr(x_u,x_{-u})]=&\mathbb{E}_{x_{-u_{out}}\sim\Delta_{-u_{out}}}[U_{u_{out}}(x_{u_{out}}=x_u,x_{-u_{out}})]\\
\geq& \mathbb{E}_{x_{-u_{out}}\sim\Delta_{-u_{out}}}[U_{u_{out}}(x_{u_{out}}^\pr=x_u^\pr,x_{-u_{out}})] -\varepsilon\\
\geq&\mathbb{E}_{x_{-u}\sim\Delta_{-u}}[U_u(x_u^\pr,x_{-u})] -\varepsilon
\end{align*}
Hence, given any $\varepsilon$-Nash equilibrium of the BNPG game with symmetric altruism ,  we can find an $\varepsilon$-Nash equilibrium of the directed public goods game in polynomial time. This concludes the proof of this theorem.
\end{proof}

\section{Results for Altruistic Network Modification}
In this section, we present the results for Altruistic Network Modification. First let us call ANM with altruism as heterogeneous ANM with altruism whenever the BNPG game is heterogeneous. Similarly let us call ANM with altruism as fully homogeneous ANM whenever the BNPG game is fully homogeneous. \cite{maiti2020parameterized} depicted a way to reduce heterogeneous BNPG game to fully homogeneous BNPG game. By extending their ideas to our setting, we show \Cref{lem:1,lem:2,lem:3} which will be helpful to prove the theorems on hardness in this section.
\begin{lemma}[$\star$]\label{lem:1}
Given an instance of heterogeneous ANM with asymmetric altruism such that cost $c_v$ of investing is same for all players $v$ in the heterogeneous BNPG game, we can reduce the instance heterogeneous ANM with asymmetric altruism to an instance of fully homogeneous ANM with asymmetric altruism.
\end{lemma}
\begin{lemma}[$\star$]\label{lem:2}
Given an instance of heterogeneous ANM with symmetric altruism such that cost $c_v$ of investing is same for all players $v$ in the heterogeneous BNPG game, we can reduce the instance heterogeneous ANM with symmetric altruism to an instance of fully homogeneous ANM with symmetric altruism.
\end{lemma}
\begin{lemma}[$\star$]\label{lem:3}
Given an instance of heterogeneous ANM with symmetric altruism such that input network has maximum degree 3, cost $c_v$ of investing is same for all players $v$ in the heterogeneous BNPG game and there are three types of externality functions, we can reduce the instance heterogeneous ANM with symmetric altruism to an instance of fully homogeneous ANM with symmetric altruism such that the input network has maximum degree 13.
\end{lemma}

ANM with asymmetric altruism is known to be \NPC when the input network is a clique \cite{ijcai2021-69}. We show a similar result for trees by reducing from Knapsack problem.
\begin{theorem}[$\star$]\label{thm:5}
For the target profile where all players invest, ANM with asymmetric altruism is \NPC when the input network is a tree and the BNPG game is fully homogeneous.
\end{theorem}

ANM with symmetric altruism is known to be \NPC when the input network is a clique \cite{ijcai2021-69}. We show a similar result for trees by reducing from Knapsack problem.

\begin{theorem}[$\star$]\label{thm:6}
For the target profile where all players invest, ANM with symmetric altruism is \NPC when the input network is a tree and the BNPG game is fully homogeneous.
\end{theorem}

We now show that ANM with symmetric altruism is known to be \PNPH for the parameter maximum degree of the input network even when the BNPG game is fully homogeneous. Towards that, we reduce from an instance of \tsat which is known to be \NPC \cite{berman2004approximation}. \tsat is the special case of 3-SAT where each variable $x_i$ occurs exactly twice as negative literal $\bar{x}_i$ and twice as positive literal $x_i$.

\begin{theorem}[$\star$]\label{thm:7}
For the target profile where all players invest, ANM with symmetric altruism is known to be \PNPH for the parameter maximum degree of the input network even when the BNPG game is fully homogeneous.
\end{theorem}

We complement the previous result by showing that ANM with asymmetric altruism is \FPT for the parameter maximum degree of the input graph.
\begin{theorem}\label{thm:8}
For any target profile, ANM with asymmetric altruism can be solved in time $2^{\Delta/2}\cdot n^{O(1)}$ where $\Delta$ is the maximum degree of the input graph.
\end{theorem}
\begin{proof}
\cite{ijcai2021-69} showed that solving an instance of asymmetric altruistic design is equivalent to solving $n$ different instances of Minimum Knapsack problem and each of these instances have at most $\Delta+1$ items. In Minimum Knapsack problem, we are give a set of items $1,\ldots,k$ with costs $p_1,\ldots p_k$ and weights $w_1,\ldots w_k$. The aim is to find a subset $S$ of items minimizing $\sum_{i\in S}w_i$ subject to the constraint that $\sum_{i\in S}p_i\geq P$. We assume that $\sum_{i\in[k]}p_i\geq P$ otherwise we don't have any feasible solution. Let us denote the optimal value by $OPT$. Let $W:=\sum_{i\in[k]}w_i$. Now consider the following integer linear program which we denote by ILP$_w$:
\begin{align*}
&\text{max} \quad
\sum_{i\in[k]}x_ip_i \\
&\text{s.t.} \quad
\sum_{i\in[k]}x_iw_i \leq w,\; x_i\in \{0,1\} \quad \forall i\in [k]
\end{align*}
The above integer linear program can be solved in time $2^{k/2}\cdot k^{O(1)}$ \cite{fedor2010exact}. Now observe that for all $w\geq OPT$, the optimal value $OPT_w$ of ILP$_w$ is at least $P$. Similarly for all $w< OPT$, the optimal value $OPT_w$ of the ILP$_w$ is less than $P$. Now by performing a binary search for $w$ on the range $[0,W]$ and then solving the above ILP repeatedly, we can compute $OPT$ in time $2^{k/2}\cdot \log W \cdot k^{O(1)}$. See \Cref{algo:2} for more details.

As discussed earlier, an Instance of asymmetric altruistic design is equivalent to solving $n$ different instances of Minimum Knapsack problem and each of these instances have at most $\Delta+1$ items. Hence ANM with asymmetric altruism can be solved in time $2^{\Delta/2}\cdot |x|^{O(1)}$ where $x$ is the input instance of ANM with asymmetric altruism.
\end{proof}

\begin{algorithm}[!htbp]
	\caption{Minimum Knapsack Solver} \label{algo:2}
	\begin{algorithmic}[1]
\STATE $\ell\gets 0$, $r\gets W, w\gets \lfloor\frac{\ell+r}{2}\rfloor$
\WHILE{\textbf{true}}
\STATE Solve ILP$_w$ and ILP$_{w+1}$.
\IF{$OPT_w<P$ and $OPT_{w+1}\geq P$}
\STATE \textbf{return} $w+1$
\ELSIF{$OPT_w<P$ and $OPT_{w+1}< P$}
\STATE $\ell\gets w+1$
\STATE $w\gets  \lfloor\frac{\ell+r}{2}\rfloor$
\ELSE
\STATE $r\gets w$
\STATE $w\gets  \lfloor\frac{\ell+r}{2}\rfloor$
\ENDIF
\ENDWHILE
	\end{algorithmic}
\end{algorithm}

We conclude our work by discussing about the approxibimility of ANM with symmetric altruism. \cite{ijcai2021-69} showed a $2+\varepsilon$ approximation algorithm for ANM with symmetric altruism when the target profile has all players investing. However, for arbitrary target profile they showed that ANM with symmetric altruism is \NPC when the input network is a complete graph and the budget is infinite. We show a similar result for graphs with bounded degree by reducing from \tsat.

\begin{theorem}[$\star$]\label{thm:9}
For an arbitrary target profile, ANM with symmetric altruism is known to be \PNPH for the parameter maximum degree of the input network even when the BNPG game is fully homogeneous and the budget is infinite.
\end{theorem}
\section{Conclusion and Future Work}
In this paper, we first studied the problem of deciding the existence of PSNE in BNPG game with altruism. We depicted polynomial time algorithms to decide the existence of PSNE in trees, complete graphs and graphs with bounded  We also that the problem of finding MSNE in BNPG game with altruism is \PPAD-Hard. Next we studied Altruistic Network modification. We showed that ANM with either symmetric or asymmetric altruism is \NPC for trees. We also showed that ANM with symmetric altruism is \PNPH for the parameter maximum degree whereas ANM with asymetric altruism is \FPT for the parameter maximum degree. One important research direction in ANM is to maximize the social welfare while ensuring that the target profile remains a PSNE. Another research direction is to improve the approximation algorithms of \cite{ijcai2021-69} for ANM with asymmetric altruism for trees and graphs with bounded degree. Another interesting future work is to look at other graphical games by considering altruism.

\bibliographystyle{splncs04}
\bibliography{references}

\section{Missing Proofs}

\begin{proof}[Proof of \Cref{thm:2}]
	Let $(\GG=(\VV,\EE),\HH=(\VV,\EE^\pr),(g_v)_{v\in\VV},(c_v)_{v\in\VV},a)$ be an instance of BNPG game with altruism. Let the circuit rank of \GG be $d$. W.l.o.g. let the graph \GG have $1$ connected component. First, let us compute the minimum spanning tree $\TT=(\VV,\EE_1)$ of $\GG$. Let $\EE^\prr:=\EE\setminus\EE_1$. Note that $|\EE^\prr|=d$. Let $\VV^\pr=\{v_1,v_2,\ldots,v_\el\}\subseteq\VV$ be the set of endpoints of the edges in $\EE^\prr$. Note that $|\VV^\pr|=\el \leq 2d$. Let $\EE_2=\{(u,v): \{u,v\}\in \EE_1 \text{ and } (u,v)\in \EE^\pr  \}$. For all $v\in\VV$, let $N_v^\pr$ denote the set of out-neighbours of $v$ in $\HH^\pr=(\VV,\EE_2)$. For every pair of tuples $t=(x_v^\pr)_{v\in\VV^\pr}\in\{0,1\}^\el$ and $s=(n_v^\pr)_{v\in\VV^\pr}\in\{0,1,\ldots,n\}^\el$, we do the following.
	
	\begin{enumerate}[label=(\roman*)]
		\item $\forall v\in\VV^\pr$, let $n_v^t$ be the number of neighbours of $v$ in $\GG^\pr=(\VV,\EE^\prr)$ who play $1$ in the tuple $t$. Now below we define $g^t_v$ and $c^s_v$ for all $v\in\VV$.
		
		\[
		g^t_v(x) =\begin{cases}
			g_v(x+n_v^t)& \text{if } v\in\VV^\pr\\
			g_v(x) & \text{otherwise}
		\end{cases}
		\]
		
		\[
		c^s_v =\begin{cases}
			c_v-a\sum_{u\in N_v\setminus N_v^\pr}\Delta g_u(x_u^\pr+n_u^\pr-1)\\
			\text{if } v\in\VV^\pr,x_v^\pr=1\\
			c_v-a\sum_{u\in N_v\setminus N_v^\pr}\Delta g_u(x_u^\pr+n_u^\pr)\\
			\text{if } v\in\VV^\pr,x_v^\pr=0\\
			c_v \text{ otherwise}
		\end{cases}
		\]
		\item Next we decide if there exists a PSNE $(x_v^\prr)_{v\in\VV}\in\{0,1\}^{\VV}$ for the instance $(\TT, \HH^\pr,(g^t_v)_{v\in\VV}, (c_v^s)_{v\in\VV},a)$ of BNPG game with altruism such that
		\begin{enumerate}
			\item $x_v^\prr=x^\pr_v$ for every $v\in\VV^\pr$
			\item $\forall v\in \VV^\pr$, the number of neighbours of $v$ in \GG whose strategies are $1$ in  $(x_v^\prr)_{v\in\VV}$ is $n_v^\pr$.
		\end{enumerate}
		This can be decided by a polynomial time algorithm due to \Cref{cor:tree}. We return \yes if such a PSNE exists.
	\end{enumerate}
	
	If no such PSNE exists for every choice of pair of tuples $t$ and $s$, then we return \no. The running time of our algorithm is $n^{O(d)}$. Now we show that our algorithm returns correct output for every input instance.
	
	In one direction, let there be a PSNE $\textbf{x}^*=(x_v^*)_{v\in\VV}$ for the instance $(\GG,\HH, (g_v)_{v\in\VV},(c_v)_{v\in\VV},a)$ of BNPG game with altruism. For all $v\in \VV$, let $n_v^*$ be the number of neighbours of $v$ in \GG whose strategies are $1$ in $\textbf{x}^*$. We now show that $\textbf{x}^*$ is a PSNE for the instance $(\TT,\HH^\pr, (g^t_v)_{v\in\VV}, (c_v^s)_{v\in\VV},a)$ of BNPG game with altruism where $t=(x_v^*)_{v\in\VV^\pr}$ and $s=(n_v^*)_{v\in\VV^\pr}$. Let $n_v^{\TT}$ denote the number of neighbors of $v\in\VV$ in $\TT$ whose strategies are $1$ in $\textbf{x}^*$. Due to the definition of $n^t_v$, we have $n^*_v=n^{\TT}_v+n_v^t$ for $v\in \VV^\pr$ and $n^*_v=n^{\TT}_v$ for $v\in\VV\setminus\VV^\pr$. Therefore, we have $\Delta g^t_v(n^{\TT}_v)=\Delta g_v(n^{\TT}_v+n_v^t)=\Delta g_v(n^*_v)$ for $v\in \VV^\pr$ and $\Delta g^t_v(n^{\TT}_v)=\Delta g_v(n^*_v)$ for $v\in\VV\setminus\VV^\pr$. Consider a player $v\in\VV$. If $x_v^*=1$, then $\Delta g_v(n^*_v)+a\sum_{u\in N_v}\Delta g_u(x_u^*+n^*_u-1)\geq c_v$ and hence, we have $\Delta g^t_v(n^{\TT}_v)+a\sum_{u\in N_v^\pr}\Delta g_u^t(x_u^*+n^{\TT}_u-1)\geq c_v^s$. Therefore, $v$ does not deviate from its decision of playing $1$ in $\TT$. Similarly if $x_v^*=0$, then $\Delta g_v(n^*_v)+a\sum_{u\in N_v}\Delta g_u(x_u^*+n^*_u)\leq c_v$ and hence, we have $\Delta g^t_v(n^{\TT}_v)+a\sum_{u\in N_v^\pr}\Delta g_u^t(x_u^*+n^{\TT}_u)\leq c_v^s$. Therefore, $v$ does not deviate from its decision of playing $0$ in $\TT$. Hence, $\textbf{x}^*$ is a PSNE for the instance $(\TT,\HH^\pr, (g^t_v)_{v\in\VV}, (c_v^s)_{v\in\VV},a)$ of BNPG game with altruism where $t=(x_v^*)_{v\in\VV^\pr}$ and $s=(n_v^*)_{v\in\VV^\pr}$. This implies that our algorithm will return \YES.
	
	In the other direction, let our algorithm return \YES. This means for a pair of tuples $t=(x_v^*)_{v\in\VV^\pr}$ and $s=(n_v^*)_{v\in\VV^\pr}$, we have a PSNE $(x_v^*)_{v\in\VV}$ for the instance $(\TT, \HH^\pr, (g^t_v)_{v\in\VV}, (c_v^s)_{v\in\VV},a)$ of BNPG game with altruism. We now show that $(x_v^*)_{v\in\VV}$ is a PSNE for the instance $(\GG,\HH, (g_v)_{v\in\VV},(c_v)_{v\in\VV},a)$ of BNPG game with altruism. Consider a player $v\in \VV$. If $x_v^*=1$ , then $\Delta g^t_v(n^{\TT}_v)+a\sum_{u\in N_v^\pr}\Delta g_u^t(x_u^*+n^{\TT}_u-1)\geq c_v^s$. Therefore, we have $\Delta g_v(n^*_v)+a\sum_{u\in N_v}\Delta g_u(x_u^*+n^*_u-1)\geq c_v$. Hence,  $v$ does not deviate from its decision of playing $1$ in $\GG$.  Similarly, if $x_v^*=0$, then $\Delta g^t_v(n^{\TT}_v)+a\sum_{u\in N_v^\pr}\Delta g_u^t(x_u^*+n^{\TT}_v)\leq c_v^s$. Therefore, we have $\Delta g_v(n^*_v)+a\sum_{u\in N_v}\Delta g_u(x_u^*+n^*_u)\leq c_v$. Hence,  $v$ does not deviate from its decision of playing $0$ in $\GG$. Hence,  $(x_v^*)_{v\in\VV}$ is a PSNE for the instance $(\GG,\HH, (g_v)_{v\in\VV},(c_v)_{v\in\VV},a)$ of BNPG game with altruism.
\end{proof}

\begin{proof}[Proof of \Cref{thm:3}]
	Let $0<k<|\VV|$ be an integer. First observe that if $k$ players are playing $1$, then $x_v+n_v$ for every player $v\in \VV$ is $k$. Let $\RR_0(k):=\{v: \Delta g_v(k)+a\cdot\sum_{u\in N_v}\Delta g_u(k)\leq c_v\}$ be the set of players which do not deviate from playing $0$ if their $k$ neighbours are playing $1$. Let $\RR_1(k):=\{v: \Delta g_v(k-1)+a\cdot\sum_{u\in N_v}\Delta g_u(k-1)\geq c_v\}$ be the set of players which do not deviate from playing $1$ if their $k-1$ neighbours are playing $1$. Now we claim that there is a PSNE where $k$ players are playing $1$ iff $|\RR_1(k)|\geq k$, $|\RR_0(k)|\geq |\VV|-k$ and $|\RR_0(k)\setminus \RR_1(k)|=|\VV\setminus \RR_1(k)|$.
	
	In one direction, suppose there is a PSNE $\textbf{x}^*$ with $k$ players playing $1$, then there is a subset $\mathcal{I}_1\subseteq \VV$ of at least $k$ players such that  $\forall v\in \mathcal{I}_1$, $c_v\leq \Delta g_v(k-1)+a\cdot\sum_{u\in N_v}\Delta g_u(k-1)$. Therefore, $|\RR_1(k)|\geq k$. Now in $\textbf{x}^*$ we have $|\VV|-k$ players who are playing $0$. It implies that there a subset $\mathcal{I}_2\subseteq \VV$ of at least $|\VV|-k$ players such that $\forall v\in \mathcal{I}_2$, $c_v\geq \Delta g_v(k)+a\cdot\sum_{u\in N_v}\Delta g_u(k)$. Therefore, $|\RR_0(k)|\geq |\VV|-k$. Now observe that a player $v\notin\RR_1(k)$ plays 0 in $\textbf{x}^*$ otherwise $v$ would deviate from its decision. This implies that $v\in \RR_0(k)$ as $v$ does not deviate from playing $0$ (recall, $\textbf{x}^*$ is a PSNE). And $\RR_0(k)\setminus \RR_1(k)\subseteq \VV\setminus \RR_1(k)$. Hence,$|\RR_0(k)\setminus \RR_1(k)|$ must be equal to $|\VV\setminus \RR_1(k)|$.
	
	In the other direction, let it be the case that $|\RR_1(k)|\geq k$, $|\RR_0(k)|\geq |\VV|-k$ and $|\RR_0(k)\setminus \RR_1(k)|=|\VV\setminus \RR_1(k)|$. Now let us construct a strategy profile $\textbf{x}^*=(x_v)_{v\in\VV}$ as follows. First, $\forall v\in \VV\setminus \RR_1(k)$, set $x_v=0$. Now consider a subset $\II \subseteq \RR_1(k)\cap\RR_0(k)$ such that $|\II|=|\VV|-k-|\VV\setminus \RR_1(k)|$. We always can choose such a subset \II due to the following:
	\begin{align*}
		|\RR_1(k)\cap\RR_0(k)|&=|\RR_0(k)|-|\RR_0(k)\setminus \RR_1(k)|\\
		&=|\RR_0(k)|-|\VV\setminus\RR_1(k)|\\
		&\geq |\VV|-k-|\VV\setminus \RR_1(k)|\\
		&= |\RR_1(k)|-k\geq 0
	\end{align*}
	Now, $\forall v\in \II$, set $x_v=0$. For the rest of the $k$ players who are not part of both \II and $ \VV\setminus \RR_1(k)$, we set their strategies as $1$. Now we claim that $\textbf{x}^*$ is a PSNE. A player $v$ with $x_v=1$ won't deviate as they are part of $\RR_1(k)$. Similarly, a player $v$ with $x_v=0$ won't deviate as they are part of $\RR_0(k)$. Hence $\textbf{x}^*$ is a PSNE.
	
	Having proved our claim, we now present the polynomial time algorithm. First check whether the strategy profile where all players do not invest is PSNE or not. This can checked in polynomial time. If it is not a PSNE, then check whether the strategy profile where all players invest is PSNE or not. This can also be checked in polynomial time. If it is not a PSNE, then check whether there is a $k$ such that $0<k<|\VV|$, $|\RR_1(k)|\geq k$, $|\RR_0(k)|\geq |\VV|-k$ and $|\RR_0(k)\setminus \RR_1(k)|=|\VV\setminus \RR_1(k)|$. If there is such a $k$, then there exists a PSNE otherwise there is no PSNE.
\end{proof}

\begin{proof}[Proof of \Cref{lem:1}]
Let $(\GG=(\VV=\{v_i: i\in[n]\},\EE),\HH=(\VV,\EE^\pr),(g_v)_{v\in\VV},(c_v)_{v\in\VV},a,C(.),B,\textbf{x}^*=(x_v)_{v\in\VV})$ be an instance of heterogeneous ANM with asymmetric altruism. $\forall v\in\VV$, let $c_v=c$. Using this instance, we create an instance $(\GG^\pr$=$(\VV^\pr,\EE^\prr),\HH^\pr$=$(\VV^\pr,E),(g_v^\pr)_{v\in\VV^\pr},(c_v^\pr)_{v\in\VV^\pr},a',C^\pr(.),B',$ $\textbf{x}^\pr=(x_v^\pr)_{v\in\VV^\pr})$ of fully homogeneous ANM with asymmetric altruism. First we set $a'=a$ and $B'=B$.
	
	Next we construct the graph $\GG^\pr=(\VV^\pr,\EE^\prr)$ as follows.
	\begin{align*}
		\VV^\pr &= \bigcup_{i\in[n]} V_i \cup\{u_i: i\in[n]\}, \text{where }\\
		\forall i\in[n], &\text{ }V_i = \{v_j^i: j\in[2+n(i-1)]\} \\
		\EE^\prr &= \bigcup_{i\in[n]} E_i \cup \{\{u_i,u_j\}: \{v_i,v_j\}\in\EE\}, \text{where }\\
		\forall i\in[n], &\text{ }E_i = \{\{u_i,v_j^i\}:j\in[2+n(i-1)]\}
	\end{align*}
	
	Let $f(x)=1+\lfloor\frac{x-2}{n}\rfloor$ and $h(x)=x-(2+n(f(x)-1))$. We now recursively define a function $g(.)$ as follows.
	\[
	g(x)=
	\begin{cases}
	c\cdot x & \text{if } x\in\{0,1,2\}\\
	g(x-1)+ \Delta g_{v_{f(x-1)}}(h(x-1))& \text{if }x>2 \\
	\end{cases}
	\]
	
Now for all $v\in \VV^\pr$, we set $g_v^\pr(.)=g(.)$ and $c_v^\pr=c$. Now for all $i\in[n]$, set $x_{u_i}^\pr=x_{v_i}$. For all $v\in\VV^\pr\setminus \{u_i:i\in[n]\}$, set $x_{v}^\pr=1$. Now we construct the altruistic network $\HH^\pr=(\VV^\pr,E)$  as follows:
\begin{equation*}
E=\{(u_i,u_j):i,j\in[n],(v_i,v_j)\in \EE^\pr\}
\end{equation*}
Now we define the cost function $C^\pr(.)$. For all $i,j\in[n]$ such that $(u_i,u_j)$ is allowed be added to $\HH^\pr$, $C^\pr((u_i,u_j))=C((v_i,v_j))$. For remaining edges $e$ which are allowed to be added to $\HH^\pr$, $C^\pr(e)=B+1$.
	
This completes the description of fully homogeneous ANM with asymmetric altruism. We now claim that the instance of heterogeneous ANM with asymmetric altruism is a \YES instance iff the instance of fully homogeneous ANM with asymmetric altruism is a \YES instance.

In one direction, let heterogeneous ANM with asymmetric altruism be a \YES instance. For all $i,j\in[n]$, if $(v_i,v_j)$ was added to \HH, then add $(u_i,u_j)$ to $\HH^\pr$. Similarly for all $i,j\in[n]$, if $(v_i,v_j)$ was removed from \HH, then remove $(u_i,u_j)$ from $\HH^\pr$. Now we show that $\textbf{x}^\pr$ becomes a PSNE. For all $i\in[n]$, let the number of neighbours of $v_i$ playing $1$ in $\textbf{x}^*$ be $n_{v_i}$. Then the number of neighbouts of $u_i$ playing $1$ in $\textbf{x}^\pr$ is $2+(i-1)n+n_{v_i}$. Now for all $i\in [n]$ such that $x_{u_i}^\pr=x_{v_i}=1$, we have $\Delta g(2+(i-1)n+n_{v_i})+a\sum_{u_j\in N_{u_i}}\Delta g(2+(j-1)n+n_{v_j}+x_{u_j}^\pr-1)=\Delta g_{v_i}(n_{v_i})+a\sum_{v_j\in N_{v_i}}\Delta g_{v_j}(n_{v_j}+x_{v_j}-1)\geq c$. Hence $u_i$ doesn't deviate from playing $1$. Similarly for all $i\in [n]$ such that $x_{u_i}^\pr=x_{v_i}=0$, we have $\Delta g(2+(i-1)n+n_{v_i})+a\sum_{u_j\in N_{u_i}}\Delta g(2+(j-1)n+n_{v_j}+x_{u_j}^\pr)=\Delta g_{v_i}(n_{v_i})+a\sum_{v_j\in N_{v_i}}\Delta g_{v_j}(n_{v_j}+x_{v_j})\leq c$. Hence $u_i$ doesn't deviate from playing $0$. Remaining nodes $u$ in $\GG^\pr$ don't deviate from playing $1$ as $\Delta g(0)$ and $\Delta g(1)$ is at least $c$. Hence fully homogeneous ANM with asymmetric altruism is a \YES instance.

In other direction, let fully homogeneous ANM with asymmetric altruism be a \YES instance. First observe that an edge not having both the endpoints in the set $\{u_i:i\in[n]\}$ can't be added to $\HH^\pr$ otherwise the total cost of the edges added would exceed $B$. For all $i,j\in[n]$, if $(u_i,u_j)$ was added to $\HH^\pr$, then add $(v_i,v_j)$ to $\HH$. Similarly for all $i,j\in[n]$, if $(u_i,u_j)$ was removed from $\HH^\pr$, then remove $(v_i,v_j)$ from $\HH$. Now we show that $\textbf{x}^*$ becomes a PSNE. For all $i\in[n]$, let the number of neighbours of $v_i$ playing $1$ in $\textbf{x}^*$ be $n_{v_i}$. Then the number of neighbouts of $u_i$ playing $1$ in $\textbf{x}^\pr$ is $2+(i-1)n+n_{v_i}$. Now for all $i\in [n]$ such that $x_{u_i}^\pr=x_{v_i}=1$, we have $\Delta g_{v_i}(n_{v_i})+a\sum_{v_j\in N_{v_i}}\Delta g_{v_j}(n_{v_j}+x_{v_i}-1)=\Delta g(2+(i-1)n+n_{v_i})+a\sum_{u_j\in N_{u_i}}\Delta g(2+(j-1)n+n_{v_j}+x_{u_j}^\pr-1)\geq c$. Hence $v_i$ doesn't deviate from playing $1$. Similarly for all $i\in [n]$ such that $x_{u_i}^\pr=x_{v_i}=0$, we have $\Delta g_{v_i}(n_{v_i})+a\sum_{v_j\in N_{u_i}}\Delta g_{v_j}(n_{v_j}+x_{v_j})=\Delta g(2+(i-1)n+n_{v_i})+a\sum_{u_j\in N_{u_i}} \Delta g(2+(j-1)n+n_{v_j}+x_{u_j}^\pr)\leq c$. Hence $v_i$ doesn't deviate from playing $0$. Hence heterogeneous ANM with asymmetric altruism is a \YES instance.
\end{proof}

\begin{proof}[Proof of \Cref{lem:2}]
Let $(\GG=(\VV=\{v_i: i\in[n]\},\EE),\HH=(\VV,\EE^\pr),(g_v)_{v\in\VV},(c_v)_{v\in\VV},a,C(.),B,\textbf{x}^*=(x_v)_{v\in\VV})$ be an instance of heterogeneous ANM with symmetric altruism. $\forall v\in\VV$, let $c_v=c$. Using this instance, we create an instance $(\GG^\pr$=$(\VV^\pr,\EE^\prr),\HH^\pr$=$(\VV^\pr,E),(g_v^\pr)_{v\in\VV^\pr},(c_v^\pr)_{v\in\VV^\pr},a',C^\pr(.),B',$ $\textbf{x}^\pr=(x_v^\pr)_{v\in\VV^\pr})$ of fully homogeneous ANM with symmetric altruism. First we set $a'=a$ and $B'=B$.
	
	Next we construct the graph $\GG^\pr=(\VV^\pr,\EE^\prr)$ as follows.
	\begin{align*}
		\VV^\pr &= \bigcup_{i\in[n]} V_i \cup\{u_i: i\in[n]\}, \text{where }\\
		\forall i\in[n], &\text{ }V_i = \{v_j^i: j\in[2+n(i-1)]\} \\
		\EE^\prr &= \bigcup_{i\in[n]} E_i \cup \{\{u_i,u_j\}: \{v_i,v_j\}\in\EE\}, \text{where }\\
		\forall i\in[n], &\text{ }E_i = \{\{u_i,v_j^i\}:j\in[2+n(i-1)]\}
	\end{align*}
	
	Let $f(x)=1+\lfloor\frac{x-2}{n}\rfloor$ and $h(x)=x-(2+n(f(x)-1))$. We now recursively define a function $g(.)$ as follows.
	\[
	g(x)=
	\begin{cases}
	c\cdot x & \text{if } x\in\{0,1,2\}\\
	g(x-1)+ \Delta g_{v_{f(x-1)}}(h(x-1))& \text{if }x>2 \\
	\end{cases}
	\]
	
Now for all $v\in \VV^\pr$, we set $g_v^\pr(.)=g(.)$ and $c_v^\pr=c$. Now for all $i\in[n]$, set $x_{u_i}^\pr=x_{v_i}$. For all $v\in\VV^\pr\setminus \{u_i:i\in[n]\}$, set $x_{v}^\pr=1$. Now we construct the altruistic network $\HH^\pr=(\VV^\pr,E)$  as follows:
\begin{equation*}
E=\{\{u_i,u_j\}:i,j\in[n],\{v_i,v_j\}\in \EE^\pr\}
\end{equation*}
Now we define the cost function $C^\pr(.)$. For all $i,j\in[n]$ such that $\{u_i,u_j\}$ is allowed be added to $\HH^\pr$, $C^\pr(\{u_i,u_j\})=C(\{v_i,v_j\})$. For remaining edges $e$ which are allowed to be added to $\HH^\pr$, $C^\pr(e)=B+1$.
	
This completes the description of fully homogeneous ANM with symmetric altruism. We now claim that the instance of heterogeneous ANM with symmetric altruism is a \YES instance iff the instance of fully homogeneous ANM with symmetric altruism is a \YES instance.

In one direction, let heterogeneous ANM with symmetric altruism be a \YES instance. For all $i,j\in[n]$, if $\{v_i,v_j\}$ was added to \HH, then add $\{u_i,u_j\}$ to $\HH^\pr$. Similarly for all $i,j\in[n]$, if $\{v_i,v_j\}$ was removed from \HH, then remove $\{u_i,u_j\}$ from $\HH^\pr$. Now we show that $\textbf{x}^\pr$ becomes a PSNE. For all $i\in[n]$, let the number of neighbours of $v_i$ playing $1$ in $\textbf{x}^*$ be $n_{v_i}$. Then the number of neighbouts of $u_i$ playing $1$ in $\textbf{x}^\pr$ is $2+(i-1)n+n_{v_i}$. Now for all $i\in [n]$ such that $x_{u_i}^\pr=x_{v_i}=1$, we have $\Delta g(2+(i-1)n+n_{v_i})+a\sum_{u_j\in N_{u_i}}\Delta g(2+(j-1)n+n_{v_j}+x_{u_j}^\pr-1)=\Delta g_{v_i}(n_{v_i})+a\sum_{v_j\in N_{v_i}}\Delta g_{v_j}(n_{v_j}+x_{v_j}-1)\geq c$. Hence $u_i$ doesn't deviate from playing $1$. Similarly for all $i\in [n]$ such that $x_{u_i}^\pr=x_{v_i}=0$, we have $\Delta g(2+(i-1)n+n_{v_i})+a\sum_{u_j\in N_{u_i}}\Delta g(2+(j-1)n+n_{v_j}+x_{u_j}^\pr)=\Delta g_{v_i}(n_{v_i})+a\sum_{v_j\in N_{v_i}}\Delta g_{v_j}(n_{v_j}+x_{v_j})\leq c$. Hence $u_i$ doesn't deviate from playing $0$. Remaining nodes $u$ in $\GG^\pr$ don't deviate from playing $1$ as $\Delta g(0)$ and $\Delta g(1)$ is at least $c$. Hence fully homogeneous ANM with symmetric altruism is a \YES instance.

In other direction, let fully homogeneous ANM with symmetric altruism be a \YES instance. First observe that an edge not having both the endpoints in the set $\{u_i:i\in[n]\}$ can't be added to $\HH^\pr$ otherwise the total cost of the edges added would exceed $B$. For all $i,j\in[n]$, if $\{u_i,u_j\}$ was added to $\HH^\pr$, then add $\{v_i,v_j\}$ to $\HH$. Similarly for all $i,j\in[n]$, if $\{u_i,u_j\}$ was removed from $\HH^\pr$, then remove $\{v_i,v_j\}$ from $\HH$. Now we show that $\textbf{x}^*$ becomes a PSNE. For all $i\in[n]$, let the number of neighbours of $v_i$ playing $1$ in $\textbf{x}^*$ be $n_{v_i}$. Then the number of neighbouts of $u_i$ playing $1$ in $\textbf{x}^\pr$ is $2+(i-1)n+n_{v_i}$. Now for all $i\in [n]$ such that $x_{u_i}^\pr=x_{v_i}=1$, we have $\Delta g_{v_i}(n_{v_i})+a\sum_{v_j\in N_{v_i}}\Delta g_{v_j}(n_{v_j}+x_{v_i}-1)=\Delta g(2+(i-1)n+n_{v_i})+a\sum_{u_j\in N_{u_i}}\Delta g(2+(j-1)n+n_{v_j}+x_{u_j}^\pr-1)\geq c$. Hence $v_i$ doesn't deviate from playing $1$. Similarly for all $i\in [n]$ such that $x_{u_i}^\pr=x_{v_i}=0$, we have $\Delta g_{v_i}(n_{v_i})+a\sum_{v_j\in N_{u_i}}\Delta g_{v_j}(n_{v_j}+x_{v_j})=\Delta g(2+(i-1)n+n_{v_i})+a\sum_{u_j\in N_{u_i}} \Delta g(2+(j-1)n+n_{v_j}+x_{u_j}^\pr)\leq c$. Hence $v_i$ doesn't deviate from playing $0$. Hence heterogeneous ANM with symmetric altruism is a \YES instance.
\end{proof}

\begin{proof}[Proof of \Cref{lem:3}]
Let $(\GG=(\VV=\{v_i: i\in[n]\},\EE),\HH=(\VV,\EE^\pr),(g_v)_{v\in\VV},(c_v)_{v\in\VV},a,C(.),B,\textbf{x}^*=(x_v)_{v\in\VV})$ be an instance of heterogeneous ANM with symmetric altruism. Let us partition \VV into three sets $\VV_1$, $\VV_2$ and $\VV_3$ such that $\forall i\in[3]$ and $\forall v\in\VV_i$, we have $g_v=g_i$ and $c_v=c$. Let $p:[n]\longrightarrow [3]$ be a function such that $p(j)=i$ if $v_j\in \VV_i$.

 Using this instance, we create an instance $(\GG^\pr$=$(\VV^\pr,\EE^\prr),\HH^\pr$=$(\VV^\pr,E),(g_v^\pr)_{v\in\VV^\pr},(c_v^\pr)_{v\in\VV^\pr},a',C^\pr(.),B',$ $\textbf{x}^\pr=(x_v^\pr)_{v\in\VV^\pr})$ of fully homogeneous ANM with symmetric altruism. First we set $a'=a$ and $B'=B$.
	
	Next we construct the graph $\GG^\pr=(\VV^\pr,\EE^\prr)$ as follows.
	\begin{align*}
		\VV^\pr &= \bigcup_{i\in[3]}V_i \cup\{u_i: i\in[n]\}, \text{where }\\
		\forall i\in[3], &\text{ }V_i = \{v_j^i: j\in[2+4(i-1)]\} \\
		\EE^\prr &= \bigcup_{i\in[3]} E_i \cup \{\{u_i,u_j\}: \{v_i,v_j\}\in\EE\}, \text{where }\\
		\forall i\in[3], &\text{ }E_i = \{\{u,v_j^i\}:j\in[2+4(i-1)],u\in V_i\}
	\end{align*}
	It is easy to observe that the maximum degree of $\GG^\pr$ is 13.

	Let $f(x)=1+\lfloor\frac{x-2}{4}\rfloor$ and $h(x)=x-(2+4(f(x)-1))$. We now recursively define a function $g(.)$ as follows.
	\[
	g(x)=
	\begin{cases}
	c\cdot x & \text{if } x\in\{0,1,2\}\\
	g(x-1)+ \Delta g_{{f(x-1)}}(h(x-1))& \text{if }x>2 \\
	\end{cases}
	\]
	For all $i\in[n]$, let $c_{u_i}^\pr=c_{v_i}$. For all $v\in\VV^\pr\setminus \{u_i:i\in[n]\}$, $c_v^\pr=c$. Now for all $v\in \VV^\pr$, we set $g_v^\pr(.)=g(.)$. Now for all $i\in[n]$, set $x_{u_i}^\pr=x_{v_i}$. For all $v\in\VV^\pr\setminus \{u_i:i\in[n]\}$, set $x_{v}^\pr=1$. Now we construct the altruistic network $\HH^\pr=(\VV^\pr,E)$  as follows:
\begin{equation*}
E=\{\{u_i,u_j\}:i,j\in[n],\{v_i,v_j\}\in \EE^\pr\}
\end{equation*}
Now we define the cost function $C^\pr(.)$. For all $i,j\in[n]$ such that $\{u_i,u_j\}$ is allowed be added to $\HH^\pr$, $C^\pr(\{u_i,u_j\})=C(\{v_i,v_j\})$. For remaining edges $e$ which are allowed to be added to $\HH^\pr$, $C^\pr(e)=B+1$.
	
This completes the description of fully homogeneous ANM with symmetric altruism. We now claim that the instance of heterogeneous ANM with symmetric altruism is a \YES instance iff the instance of fully homogeneous ANM with symmetric altruism is a \YES instance.

In one direction, let heterogeneous ANM with symmetric altruism be a \YES instance. For all $i,j\in[n]$, if $\{v_i,v_j\}$ was added to \HH, then add $\{u_i,u_j\}$ to $\HH^\pr$. Similarly for all $i,j\in[n]$, if $\{v_i,v_j\}$ was removed from \HH, then remove $\{u_i,u_j\}$ from $\HH^\pr$. Now we show that $\textbf{x}^\pr$ becomes a PSNE. For all $i\in[n]$, let the number of neighbours of $v_i$ playing $1$ in $\textbf{x}^*$ be $n_{v_i}$. Then the number of neighbouts of $u_i$ playing $1$ in $\textbf{x}^\pr$ is $2+4(p(i)-1)+n_{v_i}$. Now for all $i\in [n]$ such that $x_{u_i}^\pr=x_{v_i}=1$, we have $\Delta g(2+4(p(i)-1)+n_{v_i})+a\sum_{u_j\in N_{u_i}}\Delta g(2+4(p(j)-1)+n_{v_j}+x_{u_j}^\pr-1)=\Delta g_{p(i)}(n_{v_i})+a\sum_{v_j\in N_{v_i}}\Delta g_{p(j)}(n_{v_j}+x_{v_j}-1)\geq c$. Hence $u_i$ doesn't deviate from playing $1$. Similarly for all $i\in [n]$ such that $x_{u_i}^\pr=x_{v_i}=0$, we have $\Delta g(2+4(p(i)-1)+n_{v_i})+a\sum_{u_j\in N_{u_i}}\Delta g(2+4(p(j)-1)+n_{v_j}+x_{u_j}^\pr)=\Delta g_{p(i)}(n_{v_i})+a\sum_{v_j\in N_{v_i}}\Delta g_{p(j)}(n_{v_j}+x_{v_j})\leq c$. Hence $u_i$ doesn't deviate from playing $0$. Remaining nodes $u$ in $\GG^\pr$ don't deviate from playing $1$ as $\Delta g(0)$ and $\Delta g(1)$ is at least $c$. Hence fully homogeneous ANM with symmetric altruism is a \YES instance.

In other direction, let fully homogeneous ANM with symmetric altruism be a \YES instance. First observe that an edge not having both the endpoints in the set $\{u_i:i\in[n]\}$ can't be added to $\HH^\pr$ otherwise the total cost of the edges added would exceed $B$. For all $i,j\in[n]$, if $\{u_i,u_j\}$ was added to $\HH^\pr$, then add $\{v_i,v_j\}$ to $\HH$. Similarly for all $i,j\in[n]$, if $\{u_i,u_j\}$ was removed from $\HH^\pr$, then remove $\{v_i,v_j\}$ from $\HH$. Now we show that $\textbf{x}^*$ becomes a PSNE. For all $i\in[n]$, let the number of neighbours of $v_i$ playing $1$ in $\textbf{x}^*$ be $n_{v_i}$. Then the number of neighbouts of $u_i$ playing $1$ in $\textbf{x}^\pr$ is $2+4(p(i)-1)+n_{v_i}$. Now for all $i\in [n]$ such that $x_{u_i}^\pr=x_{v_i}=1$, we have $\Delta g_{p(i)}(n_{v_i})+a\sum_{v_j\in N_{v_i}}\Delta g_{p(j)}(n_{v_j}+x_{v_i}-1)=\Delta g(2+4(p(i)-1)+n_{v_i})+a\sum_{u_j\in N_{u_i}}\Delta g(2+4(p(j)-1)+n_{v_j}+x_{u_j}^\pr-1)\geq c$. Hence $v_i$ doesn't deviate from playing $1$. Similarly for all $i\in [n]$ such that $x_{u_i}^\pr=x_{v_i}=0$, we have $\Delta g_{p(i)}(n_{v_i})+a\sum_{v_j\in N_{u_i}}\Delta g_{p(j)}(n_{v_j}+x_{v_j})=\Delta g(2+4(p(i)-1)+n_{v_i})+a\sum_{u_j\in N_{u_i}} \Delta g(2+4(p(j)-1)+n_{v_j}+x_{u_j}^\pr)\leq c$. Hence $v_i$ doesn't deviate from playing $0$. Hence heterogeneous ANM with symmetric altruism is a \YES instance.
\end{proof}

\begin{proof}[Proof of \Cref{thm:5}]
	We reduce from the decision version of the KNAPSACK Problem. In Knapsack problem, we are give a set of items $1,\ldots,n$ with profits $p_1,\ldots p_n$ and weights $w_1,\ldots w_n$. The aim is to check whether there exists a subset $S$ of items such that $\sum_{i\in S}p_i\geq P$ and $\sum_{i\in S}w_i\leq W$. Now we create an instance of ANM with asymmetric altruism. We set $a=1$. The input graph \GG(\VV,\EE) is defined as follows
	\begin{align*}
		\VV &= \{u_{i}:i\in[2n+1]\}\\
		\EE &= \{\{u_{i},u_{n+i}\},\{u_{i},u_{2n+1}\}: i\in [n]\}\\
	\end{align*}
	Let the initial altruistic graph $\HH(\VV,\EE^\pr)$ be empty. Let the target profile $x^*$ have all the players investing (i.e, playing 1). Now we define the functions $g_u(.)$ for all $u\in \VV$. $g_{u_{2n+1}}(x)=0$ for all $x\geq 0$. For all $i\in[n]$, $g_{u_i}(x)=x\cdot p_i$ for all $x\geq 0$. For all $i\in[n]$, $g_{u_{n+i}}(x)=x\cdot P$ for all $x\geq 0$. For all $i\in[2n+1]$, $c_{u_i}=P$. Now we define the cost of the introducing an atruistic edge. The cost of introducing the edge $(u_{2n+1},u_{i})$ is $w_i$. For remaining edges $e$ which are allowed to be added to $\HH$, cost of adding $e$ is $0$. Let the total budget be $W$. This completes the description of the instance of ANM with asymmetric altruism.
	
	Now we show that KNAPSACK Problem is a yes instance iff ANM with asymmetric altruism is a yes instance. In other direction, let KNAPSACK problem be a yes instance. Then there is a subset $S$  of items such that $\sum_{i\in S}p_i\geq P$ and $\sum_{i\in S}w_i\leq W$. Now if we introduce the set of altruistic edges $\{(u_{2n+1},u_i):i\in S\}\cup\{(u_i,u_{n+i}):i\in[n]\}$, then the target profile becomes a PSNE and the total of introducing these edges is at most $B$. Hence ANM with asymmetric altruism is a yes instance.
	
	In the other direction, let the ANM with asymmetric altruism be a yes instance. Then there is a set of altruistic edges $S'$ of total cost at most $B$ such that when they are introduced the target profile becomes a PSNE. Let $S:=\{i:(u_{2n+1},u_i)\in S'\}$. Hence if the subset $S$  of items is chosen then we have $\sum_{i\in S}p_i\geq P$ and $\sum_{i\in S}w_i\leq W$.  Hence the KNAPSACK problem is a yes instance.
	
	Applying \Cref{lem:1} concludes the proof of this theorem.
\end{proof}

\begin{proof}[Proof of \Cref{thm:6}]
We reduce from the decision version of the KNAPSACK Problem. In Knapsack problem, we are give a set of items $1,\ldots,n$ with profits $p_1,\ldots p_n$ and weights $w_1,\ldots w_n$. The aim is to check whether there exists a subset $S$ of items such that $\sum_{i\in S}p_i\geq P$ and $\sum_{i\in S}w_i\leq W$. Now we create an instance of ANM with symmetric altruism as follows. We set $a=1$. The input graph \GG(\VV,\EE) is defined as follows
\begin{align*}
		\VV &= \{u_{i}:i\in[2n+1]\}\\
		\EE &= \{\{u_{i},u_{n+i}\},\{u_{i},u_{2n+1}\}: i\in [n]\}\\
\end{align*}
Let the initial altruistic graph $\HH(\VV,\EE^\pr)$ be empty. Let the target profile $x^*$ have all the players investing (i.e, playing 1). Now we define the functions $g_u(.)$ for all $u\in \VV$. $g_{u_{2n+1}}(x)=0$ for all $x\geq 0$. For all $i\in[n]$, $g_{u_i}(x)=x\cdot p_i$ for all $x\geq 0$. For all $i\in[n]$, $g_{u_{n+i}}(x)=x\cdot P$ for all $x\geq 0$. For all $i\in[2n+1]$, $c_{u_i}=P$. Now we define the cost of the introducing an atruistic edge. For all $i\in[n]$, the cost of introducing the edge $\{u_{2n+1},u_{i}\}$ is $w_i$. For all $i\in[n]$, the cost of introducing the edge $\{u_{n+i},u_{i}\}$ is $0$. Let the total budget be $W$. This completes the description of the instance of ANM with symmetric altruism.

Now we show that KNAPSACK Problem is a yes instance iff ANM with symmetric altruism is a yes instance. In other direction, let KNAPSACK problem be a yes instance. Then there is a subset $S$  of items such that $\sum_{i\in S}p_i\geq P$ and $\sum_{i\in S}w_i\leq W$. Now if we introduce the set of altruistic edges $\{\{u_{2n+1},u_i\}:i\in S\}\cup\{\{u_{n+i},u_i\}:i\in[n]\}$, then the target profile becomes a PSNE and the total of introducing these edges is at most $B$. Hence ANM with symmetric altruism is a yes instance.

In the other direction, let the ANM with symmetric altruism be a yes instance. Then there is a set of altruistic edges $S'$ of total cost at most $B$ such that when they are introduced the target profile becomes a PSNE. Let $S:=\{i:\{u_{2n+1},u_i\}\in S'\}$. Hence if the subset $S$  of items is chosen then we have $\sum_{i\in S}p_i\geq P$ and $\sum_{i\in S}w_i\leq W$.  Hence the KNAPSACK problem is a yes instance.

Applying \Cref{lem:2} concludes the proof of this theorem.
\end{proof}

\begin{proof}[Proof of \Cref{thm:9}]
To show the \NP-hardness, we reduce from an instance of \tsat which we denote by $(\XX=\{x_i:{i\in[n]}\}, \CC=\{C_j: j\in[m]\})$. We define a function $h:\{x_i,\bar{x}_i: i\in[n]\}\rightarrow\{z_i,\bar{z}_i: i\in[n]\}$ as $h(x_i)=z_i$ and $h(\bar{x}_i)=\bar{z}_i$ for all $i\in[n]$. We now create an instance of ANM with symmetric altruism as follows. Let the initial altruistic network \HH be empty and $a=2$. Input graph \GG=(\VV,\EE) for the input BNPG game is as follows:
	\begin{align*}
	\VV &= \{z_i, \bar{z}_i,b_i: i\in[n]\} \cup \{y_j: j\in[m]\}\\
	\EE &= \{\{y_j,h(l_1^j)\}, \{y_j,h(l_2^j)\}, \{y_j,h(l_3^j)\}:
	C_j =\\ & (l_1^j\vee l_2^j\vee l_3^j), j\in[m]\}
	\cup \{\{z_i,b_i\},\{b_i,\bar{z}_i\}: i\in[n]\}
	\end{align*}
	
	Now observe that the degree of every vertex in \GG is at most $3$. We now define $(c_v)_{v\in\VV}$ and $(g_v)_{v\in\VV}$. $\forall v\in \VV, c_v=15$. $\forall j\in[m] \text{ and } \forall x\in \NB\cup\{0\},\text{ }g_{y_j}(x)=x$. $\forall i\in[n] \text{ and } \forall x\in \NB\cup\{0\},\text{ } g_{z_i}(x)=g_{\bar{z}_i}(x)= 10x$. $\forall i\in[n] \text{ and } \forall x\in \NB\cup\{0\},\text{ } g_{b_i}(x)= 2x$. Target profile $\textbf{x}^*=(x_v)_{v\in\VV}$ is defined as follows. For all $j\in[m]$, $x_{y_j}=1$. For all $i\in[n]$, $x_{z_i}=x_{\bar{z}_i}=0$ and $x_{b_i}=1$. Now define the cost of adding edges to $\HH$. Cost of adding any edge is $0$. The budget $B$ is infinite. This completes the construction of the instance of ANM with symmetric altruism.

Now we prove that the instance of \tsat is a \YES instance iff the instance of ANM with symmetric altruism is a \YES instance. In one direction, let \tsat be a \YES instance and its satisfying assignment be $f:\{x_i,\bar{x}_i: i\in[n]\}\rightarrow\{\true, \false\}$. For all $i\in[n]$, we now do the following: 
	\begin{itemize}
		\item Let $h^{-1}(z_i)$ be part of the clauses $C_j$ and $C_{j'}$. Then add the edges $\{z_i,y_j\}, \{z_i,y_{j'}\}$ to \HH if $f(x_i)=\true$, otherwise add the edge $\{z_i,b_i\}$.
		\item Let $h^{-1}(\bar{z}_i)$ be part of the clauses $C_j$ and $C_{j'}$. Then add the edges $\{\bar{z}_i,y_j\}, \{\bar{z}_i,y_{j'}\}$ to \HH if $f(\bar{x}_i)=\true$, otherwise add the edge $\{\bar{z}_i,b_i\}$.
	\end{itemize}
It is easy to observe that after adding the above edges, $\textbf{x}^*$ becomes a PSNE. Hence the instance of ANM with symmetric altrusim is a \YES instance.
	
In the other direction, let ANM with symmetric altruism be a \YES instance. First observe that for all $i\in[n]$, either $\{z_i,b_i\}$ or $\{\bar{z}_i,b_i\}$ must have been added to \HH. Let $S:=\{u :u\in\VV, \{u,b_i\} \text{ is added to \HH }\}$. Let $S^\pr:=  \{z_i, \bar{z}_i: i\in[n]\}\setminus S$. Now for all $u\in S$, there is no $j\in[m]$ such that $\{u,y_j\}$ is added to \HH otherwise $u$ will deviate from playing $0$. Now consider the following assignment $f:\{x_i,\bar{x}_i: i\in[n]\}\rightarrow\{\true, \false\}$ of \tsat instance. For all $i\in[n]$, if $h(x_i)\in S^\pr$ we have $f(x_i)=\true$ otherwise we have $f(x_i)=\false$. Now we show that $f$ is a satisfying assignment. If not there is a clause $C_j = (l_1^j\vee l_2^j\vee l_3^j)$ which is not satisfied. Hence $\{h(l_1^j),y_j\}, \{h(l_1^j),y_j\}, \{h(l_1^j),y_j\}$ are not added to \HH. But then $y_j$ would deviate from playing $1$ which is a contradiction. Hence $f$ is a satisfying assignment. Hence the instance of \tsat is a \YES instance.

Applying \Cref{lem:3} concludes the proof of this theorem.
\end{proof}

\begin{proof}[Proof of \Cref{thm:7}]
	To show the \NP-hardness, we reduce from an instance of \tsat which we denote by $(\XX=\{x_i:{i\in[n]}\}, \CC=\{C_j: j\in[m]\})$. We define a function $h:\{x_i,\bar{x}_i: i\in[n]\}\rightarrow\{z_i,\bar{z}_i: i\in[n]\}$ as $h(x_i)=z_i$ and $h(\bar{x}_i)=\bar{z}_i$ for all $i\in[n]$. We now create an instance of ANM with symmetric altruism as follows. Let the initial altruistic network \HH be empty and $a=0.5$. Input graph \GG=(\VV,\EE) for the input BNPG game is as follows:
	\begin{align*}
		\VV &= \{z_i, \bar{z}_i,b_i: i\in[n]\} \cup \{y_j: j\in[m]\}\\
		\EE &= \{\{y_j,h(l_1^j)\}, \{y_j,h(l_2^j)\}, \{y_j,h(l_3^j)\}:
		C_j =\\ & (l_1^j\vee l_2^j\vee l_3^j), j\in[m]\}
		\cup \{\{z_i,b_i\},\{b_i,\bar{z}_i\}: i\in[n]\}
	\end{align*}
	
	Now observe that the degree of every vertex in \GG is at most $3$. We now define $(c_v)_{v\in\VV}$ and $(g_v)_{v\in\VV}$. $\forall v\in \VV, c_v=315$. $\forall j\in[m] \text{ and } \forall x\in \NB\cup\{0\},\text{ }g_{y_j}(x)=220x$. $\forall i\in[n] \text{ and } \forall x\in \NB\cup\{0\},\text{ } g_{z_i}(x)=g_{\bar{z}_i}(x)= 200x$. $\forall i\in[n] \text{ and } \forall x\in \NB\cup\{0\},\text{ } g_{b_i}(x)= 240x$. Target profile $\textbf{x}^*$ has all the players investing. Now define the cost of adding edges to $\HH$. Cost of adding an edge from the set $\{\{y_j,h(l_1^j)\}, \{y_j,h(l_2^j)\}, \{y_j,h(l_3^j)\}: C_j =(l_1^j\vee l_2^j\vee l_3^j), j\in[m]\}$ is $1$. Cost of adding an edge from the set $\{\{z_i,b_i\},\{b_i,\bar{z}_i\}: i\in[n]\}$ is $c=2n+1$. The budget $B$ is $n(2+c)$. This completes the construction of the instance of ANM with symmetric altruism.
	
	Now we prove that the instance of \tsat is a \YES instance iff the instance of ANM with symmetric altruism is a \YES instance. In one direction, let \tsat be a \YES instance and its satisfying assignment be $f:\{x_i,\bar{x}_i: i\in[n]\}\rightarrow\{\true, \false\}$. For all $i\in[n]$, we now do the following:
	\begin{itemize}
		\item Let $h^{-1}(z_i)$ be part of the clauses $C_j$ and $C_{j'}$. Then add the edges $\{z_i,y_j\}, \{z_i,y_{j'}\}$ to \HH if $f(x_i)=\true$, otherwise add the edge $\{z_i,b_i\}$.
		\item Let $h^{-1}(\bar{z}_i)$ be part of the clauses $C_j$ and $C_{j'}$. Then add the edges $\{\bar{z}_i,y_j\}, \{\bar{z}_i,y_{j'}\}$ to \HH if $f(\bar{x}_i)=\true$, otherwise add the edge $\{\bar{z}_i,b_i\}$.
	\end{itemize}
	The total cost of adding the above edges is $n(2+c)$. It is easy to observe that after adding the above edges, $\textbf{x}^*$ becomes a PSNE. Hence the instance of ANM with symmetric altruism is a \YES instance.\shortversion{Similarly, the other direction can also be shown.}
	
	In the other direction, let ANM with symmetric altruism be a \YES instance. First observe that for all $i\in[n]$, either $\{z_i,b_i\}$ or $\{\bar{z}_i,b_i\}$ must have been added to \HH. Also for no $i\in [n]$, both $\{z_i,b_i\}$ and $\{\bar{z}_i,b_i\}$ are added to \HH otherwise the total cost of edges added would exceed $B$. Let $S:=\{u :u\in\VV, \{u,b_i\} \text{ is added to \HH }\}$. Let $S^\pr:=  \{z_i, \bar{z}_i: i\in[n]\}\setminus S$. Now for all $u\in S^\pr$, $\{u,y_j\},\{u,y_{j'}\}$ must have been added to \HH otherwise $u$ will deviate from playing $1$. Here $h^{-1}(u)$ are part of the clauses $C_j$ and $C_{j'}$. No other edge can be added to \HH otherwise the total cost of edges added would exceed $B$. Now consider the following assignment $f:\{x_i,\bar{x}_i: i\in[n]\}\rightarrow\{\true, \false\}$ of \tsat instance. For all $u\in S^\pr$ we have $f(h^{-1}(u))=\true$ and for all $v\in S$ we have $f(h^{-1}(v))=\false$. Observe that there is no $i\in[n]$ such that $f(x_i)=f(\bar{x}_i)$. Now we show that $f$ is a satisfying assignment. If not there is a clause $C_j = (l_1^j\vee l_2^j\vee l_3^j)$ which is not satisfied. Hence $\{h(l_1^j),y_j\}, \{h(l_1^j),y_j\}, \{h(l_1^j),y_j\}$ are not added to \HH. But then $y_j$ would deviate from playing $1$ which is a contradiction. Hence $f$ is a satisfying assignment. Hence the instance of \tsat is a \YES instance.
	Applying \Cref{lem:3} concludes the proof of this theorem.
\end{proof}

\end{document}